\documentclass{amsart}

\usepackage{amsmath, amsthm}
\usepackage{mathtools, mathabx}
\usepackage{hyperref}

\allowdisplaybreaks

\newtheorem{theorem}{Theorem}[section]

\newtheorem{lemma}[theorem]{Lemma}

\newtheorem{conjecture}[theorem]{Conjecture}
\theoremstyle{definition}

\newtheorem*{acknowledgement}{Acknowledgement}

\theoremstyle{remark}
\newtheorem{remark}[theorem]{Remark} 

\numberwithin{equation}{section}
\numberwithin{figure}{section}

\DeclarePairedDelimiter{\abs}{\lvert}{\rvert}

\DeclarePairedDelimiter{\braket}{\langle}{\rangle}
\DeclarePairedDelimiter{\dbraket}{\langle\langle}{\rangle\rangle}
\DeclarePairedDelimiter{\normalorder}{:}{:}

\DeclareMathOperator*{\Res}{Res}
\DeclareMathOperator{\Tr}{Tr}
\DeclareMathOperator{\diag}{diag}

\newcommand*{\proj}{\mathcal{P}}
\newcommand*{\Mbar}{\widebar{\mathcal{M}}}
\newcommand*{\moduli}{\mathcal{M}}
\newcommand*{\halfgenus}{\mathfrak{h}}

\newcommand*{\LLtilde}{\skew{-5}{\widetilde}{\mathbb{L}}}
\newcommand*{\Lhat}{\widehat{L}}
\newcommand*{\Mhat}{\widehat{M}}
\newcommand*{\openL}{{\Lhat}^{\mathfrak{o}}}
\newcommand*{\openM}{{\Mhat}^{\mathfrak{o}}}

\newcommand*{\tBold}{\mathbf{t}}

\newcommand*{\CC}{\mathbb{C}} 
\newcommand*{\LL}{\mathbb{L}} 
 
\newcommand*{\PP}{\mathbb{P}}

\newcommand*{\ZZ}{\mathbb{Z}}

\newcommand*{\cC}{\mathcal{C}}
\newcommand*{\cD}{\mathcal{D}}
\newcommand*{\cE}{\mathcal{E}} 
 
\newcommand*{\cH}{\mathcal{H}} 
 
\newcommand*{\cL}{\mathcal{L}} 
\newcommand*{\cM}{\mathcal{M}} 
\newcommand*{\cR}{\mathcal{R}}

\newcommand*{\cU}{\mathcal{U}}

\newcommand*{\fo}{\mathfrak{o}}

\begin{document}

\title{Topological recursion for open intersection numbers}
\subjclass[2010]{Primary 14H15; Secondary 14H70, 53D45}
\author[B. Safnuk]{Brad Safnuk}
\address{
    Department of Mathematics\\
    Central Michigan University\\
    Mount Pleasant, MI 48859, U.S.A.}
\email{brad.safnuk@cmich.edu}

\begin{abstract}
 We present a topological recursion formula for calculating the intersection numbers defined on the moduli space of open Riemann surfaces.
 The spectral curve is $x = \frac{1}{2}y^2$, the same as spectral curve used to calculate intersection numbers for closed Riemann surfaces, but the formula itself is a variation of the usual Eynard-Orantin recursion. It looks like the recursion formula used for spectral curves of degree 3, and also includes features present in $\beta$-deformed models. The recursion formula suggests a conjectural refinement to the generating function that allows for distinguishing intersection numbers on moduli spaces with different numbers of boundary components.

 \end{abstract}

\maketitle

\section{Introduction}
In \cite{Pandharipande:2014fk}, Pandharipande, Solomon and Tessler constructed a rigorous theory of intersection theory on the moduli space of the disk, and proved that the generating function for these numbers obey a number of constraint conditions that are direct analogues of the KdV equation and Virasoro constraints for intersection theory on moduli spaces of closed Riemann surfaces
(c.f. \cite{Kontsevich:1992fv, MR1144529, MR1083914}).
These constraints uniquely specify the intersection theory for higher genera intersection numbers,  and led to conjectural equations for the resulting generating functions.
In particular, they conjectured a Virasoro constraint condition, and a solution of an integrable system that they termed the open KdV equation.
Later, it was shown by Buryak \cite{Buryak:2014kx} that these two systems of differential equations are, in fact, compatible.
In so doing, he demonstrated that the open KdV equations form a part of a larger hierarchy, called the Burgers-KdV equations, which he conjectured to be the correct framework for introducing descendent integrals for the marked points appearing on the boundary of the surfaces.
These conjectures have been proven in \cite{Buryak:2015uq,Tessler:2015ys}, with the caveat that the rigorous construction of the necessary moduli spaces has been announced by Solomon and Tessler, but as of the writing of this paper, has not yet appeared.

Alexandrov \cite{Alexandrov:2015kq} constructed a solution to the Burgers-KdV hierarchy based on the Kontsevich-Penner matrix model. The partition function for this matrix model was shown to satisfy the so-called MKP hierarchy. In addition, Alexandrov \cite{Alexandrov:2014gfa} constructed a $W$-algebra constraint on this function. 

In the present work, we take Alexandrov's $W$-constraint as the starting point, and show that it is equivalent to a topological recursion equation for the generating function of open intersection numbers. We formulate this recursion in two ways: first as a master equation, in the sense of Kazarian and Zograf \cite{Kazarian:2014ys} (also present in \cite{Eynard:2007kx}), and then as a residue calculation.

The spectral curve turns out to be the same as the spectral curve for the Witten-Kontsevich generating function (c.f. \cite{MR2944483, Bergere:2009ve}), namely
\begin{align*}
    x &= \frac{1}{2}z^2 \\
    y &= z,
\end{align*}
however, the topological recursion formula itself is rather unusual, in that it combines aspects of the topological recursion formula for curves with higher order branch points 
(c.f. \cite{Bouchard:2013ss, MR3147410, Ferrer:2010fk}), as well as $\beta$-deformed topological recursion, as constructed in \cite{Bergere:2012dp, MR3165804}.

We also explore a potential refinement of the open intersection numbers by incorporating a grading parameter $Q$, which separates the components of a given moduli space by the number of boundary components in the underlying surface.
We conjecture that the resulting generating function is given by a parametrized version of the Kontsevich-Penner matrix model. If true, it would immediately imply a quantum curve equation for the principal specialization of the generating function:
\begin{equation*}
    \left(\hbar^3 \frac{d^3}{dx^3} - 2\hbar x \frac{d}{dx} + 2\hbar(Q-1)\right) e^{\Psi_Q} = 0.
\end{equation*}
Note that the quantum curve for open intersection numbers (without the $Q$ grading) is given by substituting $Q=1$.

The paper is organized as follow.
In Section~\ref{sect:W-Constraint}, we review the recent results on open intersection numbers and, in particular, Alexandrov's construction of $W$-constraints on the generating function of open intersection numbers.
Then, in Section~\ref{sect:TR} we introduce the tools and notation used in topological recursion calculations.
In Section~\ref{sect:MasterEquation} we formulate the $W$-constraint condition as an equivalent system of master equations.
In Section~\ref{sect:Residue} we write the master equation as a residue integral.
In Section~\ref{sect:Q-Grading} we conjecture a refinement of the generating function which tracks the intersection numbers for moduli spaces of surfaces with different numbers of boundary components.
Finally, in Section~\ref{sect:QuantumCurve}, we derive the quantum curve for the $Q$-graded correlators.

\begin{acknowledgement}
    The author thanks Bertrand Eynard for useful discussions. During the preparation of this paper, the author received support from the Max Planck Institute for Mathematics in Bonn, the Institute des Hautes \'Etudes Scientifique, and the National Science Foundation through grants 1002477 and DMS-1308604. 
\end{acknowledgement}

\section{Open intersection numbers and Alexandrov's W-constraints}
\label{sect:W-Constraint}

There are many effective methods for computing so-called descendent integrals on moduli spaces of closed Riemann surfaces:
\begin{equation*}
    \braket{\tau_{a_1} \cdots \tau_{a_l}}_g = \int_{\Mbar_{g,l}} \psi_1^{a_1} \cdots \psi_{l}^{a_l}, 
\end{equation*}
where $\psi_i$ is the first chern class of the natural line bundle on $\Mbar_{g,l}$ given by taking the cotangent space of the curve at the $i$-th marked point. The approach pioneered by Witten is to collect them into generating functions
\begin{align*}
    F_g^{\text{WK}} (T_0, T_1, \ldots) &= \sum_{n=0}^{\infty} \frac{1}{n!}\braket{\gamma^n}_g, \\
    F^{\text{WK}} &= \sum_{g=0}^{\infty} u^{2g-2} F_g^{\text{WK}}, \\
    \tau_{\text{WK}} & = e^{F^{\text{WK}}},
\end{align*}
for $\gamma = \sum T_k \tau_k$.
He conjectured \cite{MR1144529}, (proven by Kontsevich \cite{Kontsevich:1992fv}) that $\tau_{\text{WK}}$ is a $\tau$-function for the KdV hierarchy, with KP times $\{t_k\}$ given by $T_k = (2k+1)!! t_{2k+1}$.
Equivalently, there is a family of differential operators $L_a^{\text{WK}}$ ($a \geq -1$) that annihilate the generating function:
\begin{equation*}
    L^{\text{WK}}_n \tau_{\text{WK}} = 0,
\end{equation*}
and satisfy the commutation relations of (half of) the Virasoro algebra
$[L^{\text{WK}}_n, L^{\text{WK}}_m] = (n-m) L^{\text{WK}}_{n+m}$.

In \cite{Pandharipande:2014fk}, Pandharipande, Solomon, and Tessler began extending descendent integration to moduli spaces of open Riemann surfaces, or more precisely, Riemann surfaces with boundary.
In particular, a Riemann surface with boundary $(X, \partial X)$ is a 1-dimensional complex manifold with finitely many circular boundaries, each with a holomorphic collar structure.
The double of $(X, \partial X)$, denoted $D(X, \partial X)$, is the closed Riemann surface obtained by Schwarz reflection across the boundary of $X$.
If $X$ has genus $g$ and $b$ boundary components then we define the augmented genus of $X$, 
$\halfgenus(X) = g + b/2$.
The genus of the double of $X$ is given by $2\halfgenus(X) - 1$.
We define $\moduli_{h, k, l}$ to be the moduli space of (possibly open) Riemann surfaces $X$ with $\halfgenus(X)= h$, $k$ marked points on the boundary of $X$, and $l$ interior marked points. We note the slightly different conventions then those used by \cite{Pandharipande:2014fk}, where in particular we use parameter $\halfgenus(X)$ instead of the genus of the doubled surface, and we consider the moduli space of closed Riemann surfaces $\moduli_{h,l}$ to be a connected component of $\moduli_{h, 0, l}$. It is also worth pointing out that $h$ can be any non-negative integer or half-integer.

$\moduli_{h, k, l}$ is a real orbifold of dimension $6h - 6 + k + 2l$. At interior marked points we have cotangent line bundles
\begin{equation*}
\LL_i \rightarrow \Mbar_{h, k, l} \quad\quad i=1, \ldots, l,
\end{equation*}
and we wish to consider $\psi_i = c(\LL_i) \in H^2(\Mbar_{h,k,l})$. Furthermore, one can construct cotangent line bundles
\begin{equation*}
{\LLtilde}_j \rightarrow \Mbar_{h, k, l} \quad\quad j=1, \ldots, k
\end{equation*}
for the marked points on the boundary and consider $\phi_j = c({\LLtilde}_j) \in H^2(\Mbar_{h, k, l})$.

If such constructions can be made rigorous, then one could calculate descendent integrals
\begin{equation}
    \braket{\tau_{a_1} \ldots \tau_{a_l} \sigma_{b_1} \cdots \sigma_{b_k}}_h^{\fo}
    = \int_{\Mbar_{h,k,l}} \psi_1^{a_1} \cdots \psi_l^{a_l} \phi_1^{b_1} \cdots \phi_k^{b_k},
\end{equation}
and the resulting generating functions
\begin{align*}
    F_h(T_0, T_1, \ldots ; S_0, S_1, \ldots) &= 
    \sum_{k,l=0}^{\infty} \frac{1}{k! l!} \braket{\gamma^k \lambda^l}_h^\fo \\
    F &= \sum_{2h \in \ZZ_{\geq 0}} u^{2h-2} F_h \\
    \tau_\fo &= e^{F},
\end{align*}
where $\gamma = \sum T_j \tau_j$ and $\lambda = \sum S_j \sigma_j$.
Although it is not yet possible to define open intersection numbers in full generality, \cite{Pandharipande:2014fk} contains a rigorous treatment for $h=0, \frac{1}{2}$ and $S_i = 0 \forall i \geq 1$, while the construction for arbitrary $h$ and $S_i$ has been announced by Solomon and Tessler. 

Based on their analysis of descendent integrals on the disk, Pandharipande, Solomon and Tessler \cite{Pandharipande:2014fk} conjectured a Virasoro constraint and open KdV equation for the generating function of open intersection numbers. In particular, for $Z = \tau_\fo(T_0, T_1, \ldots, S_0=S, S_1=0, S_2=0, \ldots)$ they conjectured that
\begin{equation*}
    \cL_n Z = 0 \quad \text{for $n=-1, 0, \ldots$},
\end{equation*}
where
\begin{multline*}
    \cL_n = - \frac{(2n+3)!!}{2^{n+1}} \frac{\partial}{\partial T_{n+1}}
    + \sum_{a=0}^{\infty} \frac{(2(a+n) + 1)!!}{2^{n+1} (2a-1)!!} T_a \frac{\partial}{\partial T_{a+n}} \\
    + \frac{u^2}{2^{n+1}} \sum_{a+b = n-1} (2a+1)!! (2b+1)!! \frac{\partial^2}{\partial T_a \partial T_b}
    + \delta_{k, -1} u^{-2} \frac{T_0^2}{2} + \delta_{n, 0} \frac{1}{16} \\
    + u^n S \frac{\partial^{n+1}}{\partial S^{n+1}} 
    + \frac{3n + 3}{4}u^n \frac{\partial^n}{\partial S^n}
\end{multline*}
satisfy commutation relations $[\cL_n, \cL_m] = (n-m)\cL_{n+m}$.

In addition, they conjectured that the generating function satisfies the following open KdV equations
\begin{multline*}
    (2n+1) \dbraket{\tau_n}^\fo = u \dbraket{\tau_{n-1}\tau_0} \dbraket{\tau_0}^\fo
    + 2 \dbraket{\tau_{n-1}}^\fo \dbraket{\sigma_0}^\fo \\
    + 2 \dbraket{\tau_{n-1}\sigma_0}^\fo
    - \frac{u}{2}\dbraket{\tau_{n-1} \tau_0^2},
\end{multline*}
where
\begin{align*}
    \dbraket{\tau_{a_1} \dots \tau_{a_l}} &= \frac{\partial}{\partial T_{a_1}} \cdots
    \frac{\partial}{\partial T_{a_l}} F^{\text{WK}}(T_0, T_1, \ldots) \\
    \dbraket{\tau_{a_1} \dots \tau_{a_l} \sigma_0^k}^\fo &= \frac{\partial}{\partial T_{a_1}} \cdots
    \frac{\partial}{\partial T_{a_l}} \frac{\partial^k}{\partial S^k}
    F^{\fo}(T_0, T_1, \ldots; S_0=S, S_1=0, S_2=0, \ldots).
\end{align*}
Many details of the proofs of these conjectures are in \cite{Buryak:2015uq, Tessler:2015ys}, with the remaining part (chiefly the construction of the compact moduli spaces and extensions of the line bundles to the boundary) announced by Solomon and Tessler.

Independent of the conjectures themselves, Buryak proved \cite{Buryak:2014kx} that the Virasoro constraint and open KdV equations are consistent and compatible with each other.
In so doing, he introduced an extended  $\tau$-function with additional parameters $S_1,S_2, \ldots$, that satisfied the Burgers-KdV hierarchy. He conjectured that these parameters introduce descendent integration with respect to the cotangent bundles from the boundary marked points.

Using the solution of the Burgers-KdV hierarchy as the starting point, Alexandrov \cite{Alexandrov:2015kq, Alexandrov:2014gfa} related the generating function of open intersection numbers to the Kontsevich-Penner matrix model
\begin{equation}
    \label{eqn:MatrixIntegral}
    \tau_Q = \det(\Lambda)^Q \cC^{-1}
    \int_{\cH_N} [d\Phi] \exp \left(
        -\Tr \Bigl( \frac{\Phi^3}{3!} - \frac{\Lambda^2\Phi}{2}
        + Q \log \Phi \Bigr)
    \right),
\end{equation}
where $Q \in \ZZ$, $\Lambda = \diag(\lambda_1, \ldots, \lambda_N)$, the integral is over the space of hermitian $N \times N$ matrices, and
\begin{equation*}
    \cC = e^{\Tr \Lambda / 3} \int d\Phi \exp \left(-\Tr \frac{\Lambda \Phi^2}{2} \right).
\end{equation*}
In particular, he conjectured that $\tau_{Q=1} = \tau_\fo$. Note that $\tau_{Q=0} = \tau_{\text{WK}}$.
In Section~\ref{sect:Q-Grading}, we formulate a conjectural relationship between $\tau_Q$ and a $Q$-graded refinement of the open intersection generating function. 

Using standard matrix model techniques, Alexandrov constructed families of operators that annihilate $\tau_1$, and satisfy the commutation relations of the $W^{(3)}$ algebra. 
He also showed that $\tau_1$ is a $\tau$-function for the KP-hierarchy, while the parameter $Q$ plays the role of a discrete time, making $\tau_Q$ a solution of the modified KP (MKP) hierarchy (c.f. \cite{Alexandrov:2013fj,MR723457}). Note that the natural KP times $\{t_k\}$ are related to the parameters $\{T_i\}$, $\{ S_i\}$ for the generating function by
\begin{align*}
T_i &= (2i+1)!! t_{2i+1} \\
    S_i &= 2^{i+1} (i+1)! t_{2i+2}.
\end{align*}

For the $W$-constraints of $\tau_1$, we define
\begin{equation}
    \label{eqn:LkOpenHat}
    \openL_k = L_{2k} + (k+2)J_{2k} + \delta_{k,0} 
    ( \frac{1}{8} + \frac{3}{2}) - J_{2k+3},
\end{equation}
and
\begin{multline}
    \label{eqn:MkOpenHat}
    \openM_k = M_{2k} + 2(k+3)L_{2k} - 2L_{2k+3} - 2(k+3)J_{2k+3} \\
    + (\frac{95}{12} + 6k + \frac{4}{3}k^2)J_{2k} + \frac{23}{4}\delta_{k,0}
    + J_{2k+6},
\end{multline}
where 
\begin{equation}
    \label{eqn:J-operator}
    J_{k} =
    \begin{cases}
        u^{1-a/3}\frac{\partial}{\partial t_k} & \text{ if $k>0$} \\
        0 & \text{ if $k=0$} \\
        (-k)u^{a/3-1} t_k & \text{ if $k<0$,}
    \end{cases}
\end{equation}
and $L_k$ and $M_k$ are the standard generators of the Virasoro and $W^{(3)}$-algebras respectively. Namely
\begin{equation}
    \label{eqn:L-operator}
    L_k = \frac{1}{2}\sum_{a+b=k} \normalorder{J_a J_b},
\end{equation}
\begin{equation}
    \label{eqn:M-operator}
    M_k = \frac{1}{3}\sum_{a+b+c=k} \normalorder{J_a J_b J_c},
\end{equation}
and $\normalorder{A}$ is the normal ordering of operator $A$.
In particular, $\normalorder{AB} = \normalorder{BA}$, while 
\begin{equation*}
    \normalorder{J_a J_b} = \begin{cases}
        J_a J_b & \text{ if $a \leq b$ } \\
        J_b J_a & \text{otherwise.}
    \end{cases}
\end{equation*}

Then Alexandrov proved \cite{Alexandrov:2014gfa} that for all $k \geq 0$
\begin{equation}
    \label{eqn:W-constraint}
    \openL_k \tau_1 = 0 = \openM_k \tau_1.
\end{equation}
In addition, these operators satisfy the commutation relations of generators of the $W^{(3)}$-algebra:
\begin{align*}
    [\openL_k, \openL_m] &= 2(k-m)\openL_{k+m}\\
    [\openM_k, \openL_m] &= 2(k-2m) \openM_{k+m} + 4m(m+1) \openL_{k+m}.
\end{align*}

As is the case with the $W^{(3)}$-algebra in general, the commutator $[\openM_k, \openM_m]$ cannot be represented as a linear combination of generators $\openL_\ell$ and $\openM_\ell$,
but is quadratic in $\openL_\ell$.


It turns out for our purposes to be more convenient to work with the following shifted operators:
\begin{align}
    {\Lhat}_k &= \openL_k \nonumber \\
    \label{eqn:hatL}
    &= L_{2k} + (k+2)J_{2k} + \delta_{k,0} 
    ( \frac{1}{8} + \frac{3}{2}) - J_{2k+3} \\
    {\Mhat}_k &= -\openM_k + 2(k+2){\Lhat}_k \nonumber \\
    \label{eqn:hatM}
              &= -M_{2k} + 2(L_{2k+3} - L_{2k}) + 2J_{2k+3} \\
    & \quad \ + (\frac{2}{3}k^2 + 2k + \frac{1}{12})J_{2k}
    + \frac{3}{4}\delta_{k,0} - J_{2k+6} \nonumber
\end{align}

\section{Topological recursion}
\label{sect:TR}

Topological recursion, as developed by Chekhov, Eynard and Orantin \cite{MR2222762,Eynard:2007kx}, originated as a method for calculating correlation functions for matrix models.
However, it was realized to be a more general construction, valid for any spectral curve (defined below), even those not arising from matrix models.
It was subsequently found to determine a large number of interesting enumerative and geometric invariants
(c.f. \cite{MR3335006,MR2746135,Do:2012lh,Dumitrescu:2013zr,MR2855174,MR2849645,Eynard:2007if,MR3339157,MR3268770,MR2944483,Dunin-Barkowski:2014ij}).

The initial data needed to apply topological recursion consists of a \emph{spectral curve} $(C, x, y)$,
where $C$ is a Torelli marked compact Riemann surface, and $x$ and $y$ are meromorphic functions on $C$.
We require that $x$ and $y$ generate the function field of $C$ (i.e. $K(C) = \CC(x, y)$), and that $x$ and $y$ separate tangents, so that we do not have $dx(p) = 0 = dy(p)$ at any point $p\in C$. We suppose that $x$ has degree $r$, and we call the zeros of $dx$ the \emph{branch points} of the spectral curve, denoted $\{a_1, \ldots, a_d\}$. 

Given the data of a spectral curve, topological recursion allows for the construction of an infinite family of \emph{correlation functions} $W_{g,n+1}(z_0, \ldots, z_n)$, defined for any $g,n \geq 0$. With the exception of $W_{0,1}$ and $W_{0,2}$, any correlation function $W_{g,n}$ is a symmetric meromorphic $n$-differential, with poles only at the branch points.

Although topological recursion is defined for spectral curves of arbitrary degree $r$, we restrict to the case of $r\leq 3$ for simplicity, and because that is sufficient for the example at hand.
We will also make the unecessary, though simplifying assumption, that the spectral curve has genus 0.

Given a spectral curve, the base cases for the correlation function are given by
\begin{align*}
    W_{0,1}(z) &= y(z)\,dx(z) \\
    W_{0,2}(z_0, z_1) &= B(z_0, z_1) = \frac{dz_0 \, dz_1}{(z_0 - z_1)^2}.
\end{align*}
We use topological recursion to calculate the remainder of the correlation functions, which utilizes the following constructions.

Given a family of symmetric $n$-differentials $\{W_{g,n}\}$,
and $\vec{z} = (z_1, \ldots, z_n)$, we define
\begin{multline*}
    \cE^{(2)}W_{g,n+1}(v, w; \vec{z}) = W_{g-1, n+2}(v, w, \vec{z}) \\
     + \sum_{\substack{g_1+g_2 = g \\ Z_1 \sqcup Z_2 = \vec{z}}}
     W_{g_1, \abs{Z_1}+1}(v, Z_1) W_{g_2, \abs{Z_2}+1}(w, Z_2),
\end{multline*}
\begin{multline*}
    \cE^{(3)}W_{g, n+1}(w_1, w_2, w_3; \vec{z})
    = W_{g-2, n+3}(w_1, w_2, w_3, \vec{z}) \\ 
     + \sum_{\substack{g_1+g_2 = g-1  \\ Z_1 \sqcup Z_2 = \vec{z}}}
     W_{g_1, \abs{Z_1}+2}(w_1, w_2, Z_1) W_{g_2, \abs{Z_2}+1}(w_3, Z_2) \\
     + \sum_{\substack{g_1+g_2+g_3 = g  \\ Z_1 \sqcup Z_2 \sqcup Z_3 = \vec{z}}}
     \prod_{i=1}^{3}W_{g_i, \abs{Z_i}+1}(w_i, Z_i),
\end{multline*}
and
\begin{equation*}
    P_{g,n+1}^{(k)}(z, \vec{z})
    = \sum_{\{w_1, \ldots, w_k\} \subset x^{-1}(x(z))}
    \cE^{(k)}_{g,n+1}(w_1, \ldots, w_k; \vec{z}).
\end{equation*}

We also define the recursion kernel 
\begin{equation*}
    K(z_0, z) = \frac{1}{r W_{0,1}(z)^{r-1}} \int_{\zeta=0}^z B(z_0,\zeta). 
\end{equation*}

Then we have the following topological recursion equation, called the global recursion in \cite{Bouchard:2013ss}. 
\begin{theorem}
    The correlation functions for a spectral curve of degree $r$ satisfy 
\begin{equation*}
    0 = \frac{1}{2\pi i} \oint_{\Gamma} K(z_0, z)
    \sum_{m=2}^{r} W_{0,1}(z)^{r-m} P_{g, n+1}^{(m)}(z, \vec{z}),
\end{equation*}
where $\Gamma$ is a contour that encloses all the branch points $\{a_1, \ldots, a_d\}$.
\end{theorem}

\section{Master equation for open intersection numbers}
\label{sect:MasterEquation}

In this section, we reframe Alexandrov's $W$-constraint condition \eqref{eqn:W-constraint} as a \emph{master equation}, in the sense of Kazarian and Zograf \cite{KazarianMasterEqNotes,Kazarian:2014ys}. We should point out that although this approach to topological recursion was formalized by Kazarian and Zograf, the idea was already present to some degree in \cite{Eynard:2007kx}. 

The operators used to express the master equation are as follows.
\begin{align*}
    \delta^{(2)} &= \sum_{k=0}^{\infty}
    \frac{dz}{z^{2k+2}}\frac{\partial}{\partial t_{2k+1}} \\
    \delta^{(3)} &= \sum_{k=1}^{\infty}
    \frac{dz}{z^{2k+1}} \frac{\partial}{\partial t_{2k}} \\
    \delta &= \delta^{(2)} + \delta^{(3)} \\
    \tBold^{(2)} &=
    \sum_{k=0}^{\infty} (2k+1)t_{2k+1} z^{2k}dz \\
    \tBold^{(3)} &=
    \sum_{k=1}^{\infty} 2k t_{2k} z^{2k-1}dz \\
    \tBold &= \tBold^{(2)} + \tBold^{(3)}
\end{align*}
Note that given any polynomial in $t$, the operator $\delta$ produces a \emph{Laurent differential} in $z$, which is a formal expression
\begin{equation*}
    \sum_{k=-m}^{\infty} a_k(t) z^k\,dz,
\end{equation*}
for some finite $m$.
As well, we define projection operators on the space of Laurent differentials.
In particular, $\proj^{(2)}$ is the projection to
the linear span of $\{ \frac{dz}{z^{2a+2}} \}_{a=0}^{\infty}$,
while $\proj^{(3)}$ is the projection to the linear span of 
$\{ \frac{dz}{z^{2a+3}} \}_{a=0}^{\infty}$.

We then form the sum
\begin{align*}
    \cL &= \sum_{k=-1}^{\infty} u^{2k/3} \frac{dz}{z^{2k+4}}{\Lhat}_k\\
        &= \sum_{k=-1}^{\infty} u^{2k/3} \frac{dz}{z^{2k+4}}
    \left( -J_{2k+3} + (k+2)J_{2k} + L_{2k} + \frac{13}{8}\delta_{k,0} \right).
\end{align*}
Term-by-term, we find that
\begin{equation*}
    -\sum_{k=-1}^{\infty} u^{2k/3} \frac{dz}{z^{2k+4}} J_{2k+3}
    = -\sum_{k=0}^{\infty} u^{(2k-2)/3} \frac{dz}{z^{2k+2}} 
    u^{1 - (2k+1)/3} \frac{\partial}{\partial t_{2k+1}}
    = -\delta^{(2)}
\end{equation*}
\begin{align*}
    \sum_{k=-1}^{\infty} u^{2k/3} \frac{dz}{z^{2k+4}}(k+2) J_{2k}
    &= 2u^{-1} t_2 \frac{dz}{z^2} 
    + \sum_{k=1}^{\infty}(k+2) u^{2k/3} \frac{dz}{z^{2k+4}} 
    u^{1 - 2k/3} \frac{\partial}{\partial_{2k}} \\
    &= 2u^{-1} t_2 \frac{dz}{z^2} 
    + u\left(-\frac{1}{2z^2}\frac{d}{dz} + \frac{3}{2z^3}\right)
    \delta^{(3)}
\end{align*}
\begin{multline*}
    \sum_{k=-1}^{\infty} u^{2k/3} \frac{dz}{z^{2k+4}} L_{2k} = 
    u^{-2} \frac{t_1^2}{2}\frac{dz}{z^2}
    + \sum_{k=-1}^{\infty} \frac{1}{z^2 dz} \sum_{-a + b = 2k}
    \frac{dz^2}{z^{-a+b+2}}a t_a \frac{\partial}{\partial t_b} \\
    + \sum_{k=1}^{\infty} u^2 \frac{1}{2z^2\,dz} \sum_{a+b=2k}
    \frac{dz^2}{z^{a+b+2}} \frac{\partial^2}{\partial t_a \partial t_b}.
\end{multline*}
We observe that the second term of the last equality is even in $z$, with order at most $z^{-2}$.
Hence we find that
\begin{multline*}
    \cL = -\delta^{(2)} + \frac{13}{8}\frac{dz}{z^4} 
    + u^{-1} 2t_2 \frac{dz}{z^2}
    + u \left( -\frac{1}{2z^2}\frac{d}{dz} + \frac{3}{2z^3} \right)
    \delta^{(3)} \\
    + \frac{u^2}{2z^2 dz}
    \left( \bigl( \delta^{(2)} \bigr)^2 
        + \bigl(\delta^{(3)} \bigr)^2
    \right)
    + u^{-2} \frac{t_1^2 \, dz}{2z^2} \\
    + \proj^{(2)}
    \left[ \frac{1}{z^2\,dz} \bigl( \tBold^{(3)}\delta^{(3)}
            + \tBold^{(2)} \delta^{(2)} \bigr)
    \right]
\end{multline*}

In addition, we define
\begin{equation*}
    \cM = \sum_{k=-2}^{\infty} u^{2k/3 + 1} \frac{dz}{z^{2k+7}}
    {\Mhat}_k,
\end{equation*}
which, by a similar calculation as was done for $\cL$, reduces to
\begin{multline*}
    \cM = -\delta^{(3)} 
    + \frac{dz}{z^3} (2u^{-1}t_1 - 4u^{-1} t_2^2 - 6u^{-1}t_1t_3
    - 5t_4 - 2u^{-2}t_1^2 t_2 ) \\
    + \frac{dz}{z^5}
    \left(- \frac{5}{2}t_2 - u^{-1}t_1^2 \right)
    + \frac{3u\,dz}{4z^7}
    + \frac{2u}{z^3}\delta^{(2)}
    + \frac{u^2}{6z^4} 
    \left( \frac{d^2}{dz^2} - \frac{3}{z}\frac{d}{dz} - \frac{9}{2z^2}
    \right) \delta^{(3)} \\
    + \frac{u^2}{z^2 dz}
    \left( \delta^{(2)} \delta^{(3)} + \delta^{(3)} \delta^{(2)}
    \right)
    + \proj^{(3)} \left[ \frac{2}{z^2\,dz} \bigl(
            \tBold^{(2)}\delta^{(3)} + \tBold^{(3)} \delta^{(2)} \bigr)
    \right] \\
    - \frac{u^3}{z^5\,dz} \left( \bigl(\delta^{(2)} \bigr)^2
    + \bigl( \delta^{(3)} \bigr)^2 \right)
    - \proj^{(3)} \left[ \frac{2u}{z^5\,dz} \bigl(
            \tBold^{(2)} \delta^{(2)} + \tBold^{(3)} \delta^{(3)} \bigr)
    \right] \\
    - \proj^{(3)} \left[
            \frac{1}{z^4\,dz^2} \bigl(2 \tBold^{(2)} \tBold^{(3)}\delta^{(2)}
            + \tBold^{(2)} \tBold^{(2)} \delta^{(3)} 
    + \tBold^{(3)} \tBold^{(3)}\delta^{(3)} \bigr)
    \right]  \\
    - \proj^{(3)} \left[
              \frac{u^2}{z^4\,dz^2} \bigl( \tBold^{(2)} \delta^{(2)} \delta^{(3)}
            + \tBold^{(2)}\delta^{(3)}\delta^{(2)}
            + \tBold^{(3)} \delta^{(2)} \delta^{(2)}
            + \tBold^{(3)} \delta^{(3)} \delta^{(3)} \bigr)
    \right] \\
    - \frac{u^4}{3 z^4\, dz^2} \biggl( \delta^{(2)}\delta^{(2)}\delta^{(3)}
            + \delta^{(2)} \delta^{(3)} \delta^{(2)}
            + \delta^{(3)} \delta^{(2)} \delta^{(2)}
            + \delta^{(3)} \delta^{(3)} \delta^{(3)} \biggr)
\end{multline*}

If we define $\cU^{(i)} = \delta^{(i)}F$, for $i=1,2$,
then the fact that $\cL e^F = 0 = \cM e^F$ implies 
\begin{multline*}
    \cU^{(2)} = \frac{13}{8}\frac{dz}{z^4} + 2u^{-1}t_2 \frac{dz}{z^2}
    + u^{-2} t_1^2\frac{dz}{2z^2} 
    + u \left( -\frac{1}{2z^2}\frac{d}{dz} + \frac{3}{2z^3}
    \right) \cU^{(3)} \\
    + \frac{u^2}{2z^2\,dz}
    \left( \bigl( \cU^{(2)} \bigr)^2 + \bigl( \cU^{(3)} \bigr)^2
        + \delta^{(2)}\cU^{(2)} + \delta^{(3)} \cU^{(3)}
    \right) \\
    + \proj^{(2)} \left[
            \frac{1}{z^2\,dz} \bigl( \tBold^{(2)} \cU^{(2)}
            + \tBold^{(3)} \cU^{(3)} \bigr)
    \right],
\end{multline*}
and
\begin{multline*}
    \cU^{(3)} = \frac{3}{4} u \frac{dz}{z^7} 
    - \left( 5t_4 \frac{dz}{z^3} + \frac{5}{2}t_2 \frac{dz}{z^5} \right)
    + u^{-1} \left( (2t_1 - 4t_2^2 - 6t_1t_3) \frac{dz}{z^3}
    - t_1^2 \frac{dz}{z^5} \right) \\
    - 2u^{-2} t_1^2 t_2 \frac{dz}{z^3} 
    + \frac{2u}{z^3}\cU^{(2)}
    + \frac{u^2}{6z^4} \left( \frac{d^2}{dz^2} - \frac{3}{z}\frac{d}{dz}
    - \frac{9}{2z^2} \right) \cU^{(3)} \\
    + \frac{u^2}{z^2\,dz} \bigl( 2\cU^{(2)} \cU^{(3)} 
    + \delta^{(2)}\cU^{(3)} + \delta^{(3)}\cU^{(2)} \bigr) \\
    - \frac{u^3}{z^5\,dz} \left(
        \bigl(\cU^{(2)} \bigr)^2
        + \bigl( \cU^{(3)} \bigr)^2 + \delta^{(2)}\cU^{(2)}
        + \delta^{(3)} \cU^{(3)}
    \right) \\
    - \frac{u^4}{3z^4\,dz^2} \biggl[
        \bigl( \delta^{(3)^2} + \delta^{(2)^2} \bigr) \cU^{(3)}
        + \bigl( \delta^{(2)}\delta^{(3)} + \delta^{(3)}\delta^{(2)}
        \bigr) \cU^{(2)} + \bigl(\cU^{(3)}\bigr)^3 \\
        + 3\bigr(\cU^{(2)}\bigl)^2 \cU^{(3)}  
        + 3 \delta^{(2)}\bigl( \cU^{(2)}\cU^{(3)} \bigr)
        + \frac{3}{2} \delta^{(3)}\bigl( {\cU^{(2)}}^2 + {\cU^{(3)}}^2 \bigr)
    \biggr] \\
    + \proj^{(3)} \biggl[
        \frac{2}{z^2\,dz} \bigl( \tBold^{(2)} \cU^{(3)}
        + \tBold^{(3)} \cU^{(2)} \bigr) 
        - \frac{2u}{z^5\,dz} \bigl( \tBold^{(2)} \cU^{(2)}
        + \tBold^{(3)}\cU^{(3)} \bigr)  \\
        - \frac{1}{z^4\,dz^2} \Bigl(
        2\tBold^{(2)} \tBold^{(3)}\cU^{(2)}
        + \bigl(\tBold^{(2)}\bigr)^2\cU^{(3)}
        + \bigl(\tBold^{(3)}\bigr)^2 \cU^{(3)} \Bigr) \\
        - \frac{u^2}{z^4\,dz^2} \Bigl(
        \tBold^{(2)} \bigl( 2\cU^{(2)} \cU^{(3)}
        + \delta^{(2)}\cU^{(3)} + \delta^{(3)}\cU^{(2)} \bigr) \\
        + \tBold^{(3)} \bigl(
        \bigl(\cU^{(2)}\bigr)^2 + \bigl(\cU^{(3)}\bigr)^2
        + \delta^{(2)}\cU^{(2)} + \delta^{(3)}\cU^{(3)} \bigr)
        \Bigr)
    \biggr].
\end{multline*}

These equations simplify substantially if we introduce
\begin{equation*}
    \tilde{\cU} = \cU - u^{-2}z^2\,dz + u^{-2}\tBold + u^{-1} \frac{dz}{z},
\end{equation*}
with $\cU = \cU^{(2)} + \cU^{(3)}$. In fact, we have
\begin{multline*}
    \proj^{(2)} \left[ \frac{u^2}{2z^2\,dz} \bigl((\tilde{\cU} - \cU)^2 
    + u^{-1}dz (- \frac{d}{dz} + \frac{1}{z})(\tilde{\cU} - \cU) \bigr) \right]  \\
    = \frac{u^{-2} t_1^2}{2z^2}dz + \frac{2u^{-1}t_2}{z^2}dz + \frac{3dz}{2z^4}
\end{multline*}
and
\begin{multline*}
    \proj^{(3)} \left[ -\frac{u^4}{3z^4\,dz^2} \bigl(
            (\tilde{\cU} - \cU)^3 
            - u^{-2} \frac{dz^2}{2}(\frac{d^2}{dz^2} - \frac{3}{z}\frac{d}{dz} + \frac{3}{2z^2} )
            (\tilde{\cU} - \cU) 
    \bigr) \right] \\
    = \Biggl\{ \frac{3}{4}u \frac{dz}{z^7} 
    - \left( 5 t_4 \frac{dz}{z^3} + \frac{5}{2}t_2 \frac{dz}{z^5} \right)
    + u^{-1} \left( (2t_1 - 6t_1 t_3 - 4 t_2^2) \frac{dz}{z^3} - t_1^2 \frac{dz}{z^5} \right) \\
    - 2u^{-2} t_1^2 t_2 \frac{dz}{z^3}\Biggr\}.
\end{multline*}
This allows us to write
\begin{equation*}
    \proj^{(2)} \left[  \frac{u^2}{2z^2\,dz} \biggl( \tilde{\cU}^2 + \delta \cU 
            + u^{-1} dz \left(- \frac{d}{dz} + \frac{1}{z} \right) \tilde{\cU} \biggr)
    \right]
    = -\frac{dz}{8z^4}
\end{equation*}
and
\begin{equation*}
    \proj^{(3)} \left[ -\frac{u^4}{3z^4\,dz^2} \biggl(
            \tilde{\cU}^3 + 3\tilde{\cU}\delta\cU + \delta^2\cU
            - u^{-2} \frac{dz^2}{2} \left(
            \frac{d^2}{dz^2} - \frac{3}{z}\frac{d}{dz} + \frac{3}{2z^2} \right) \tilde{\cU}
    \biggr) \right] 
    = 0.
\end{equation*}

An alternative formulation of the above master equations is obtained through ``renormalization.'' 
In effect, one redefines the indeterminate form
\begin{equation*}
    \delta \tBold = \frac{dz^2}{z^2} \sum_{k=1}^{\infty} k,
\end{equation*}
and instead sets it equal to $ \frac{dz^2}{4z^2}$. Then we have
\begin{equation*}
    \proj^{(2)} \left[ \frac{1}{2\eta} \bigl(
            \tilde{\cU}^2 + \delta\tilde{\cU} + u^{-1}\cD_1\tilde{\cU} \bigr)
    \right] = 0,
\end{equation*}
and
\begin{equation*}
    \proj^{(3)} \left[ \frac{1}{3\eta^2} \bigl( 
            \tilde{\cU}^3 + \tfrac{3}{2}\delta\tilde{\cU}^2 + \delta^2\tilde{\cU} - u^{-2}\cD_2 \tilde{\cU}
    \bigr) \right] = 0,
\end{equation*}
where
\begin{align*}
    \eta &= -z^2\,dz \\
    \cD_1 &= dz \left(-\frac{d}{dz} + \frac{1}{z}\right) \\
    \cD_2 &= \frac{dz^2}{2} \left( \frac{d^2}{dz^2} - \frac{3}{z}\frac{d}{dz} + \frac{3}{z^2} \right).
\end{align*}
Although the renormalization seems arbitrary, it is well justified when one transforms these master equations into an equivalent residue equation, as is done in Section~\ref{sect:Residue}.

We can make a further reduction by defining $U = \tilde{\cU} + \delta$.
Then we have proven
\begin{theorem}

\begin{align}
    \label{eq:MasterEquationSimple1}
    \proj^{(2)} \left[ 
            \frac{1}{2\eta}(U^2 + u^{-1}\cD_1 U) \cdot 1
    \right] &= 0 \\
    \label{eq:MasterEquationSimple2}
    \proj^{(3)} \left[
            \frac{1}{3\eta^2} (U^3 - u^{-2}\cD_2 U) \cdot 1
    \right] &= 0.
\end{align}

\end{theorem}

\begin{remark}
    Kazarian \cite{KazarianMasterEqNotes} has shown that the master equation for the Witten-Kontsevich $\tau$ function can be written in the form
    \begin{equation*}
        \proj^{(2)} \left[
                \frac{1}{2\eta}U^2 \cdot 1
        \right] = 0,
    \end{equation*}
    with identical initial conditions as the open intersection case, except for the lack of the term $u^{-1}\frac{dz}{z}$, and with the same renormalization condition on $\delta\tBold$. 
\end{remark}

\section{Residue formulation}
\label{sect:Residue}

In order to reformulate the master equations \eqref{eq:MasterEquationSimple1}, \eqref{eq:MasterEquationSimple2} as an example of Eynard-Orantin topological recursion \cite{Eynard:2007kx},
we must decompose the equation by genus and marked points,
convert the projection operators into residue integrals and finally,
write the equations using symmetric correlation differentials.

\begin{lemma}
    Given a Laurent differential $\gamma(z)\,dz$, we have
    \begin{align*}
        \proj^{(2)}\left[ \gamma(z)\,dz \right]
        &= \Res_{w\rightarrow 0} 
        \left[
                \frac{1}{2}\left( \frac{1}{z-w} - \frac{1}{z+w} \right)
        \gamma(w)\,dw \right] \\
        \proj^{(3)}\left[ \gamma(z)\, dz \right]
        &= \Res_{w\rightarrow 0}
        \left[
                \frac{1}{2} 
                \left(\frac{1}{z-w} - \frac{1}{z} + \frac{1}{z+w} - \frac{1}{z} \right) \gamma(w)\,dw
        \right]
    \end{align*}
    
\end{lemma}
\begin{proof}
    This is an easy consequence of the definition of a Laurent differential, and from a direct calculation of the residues against the Laurent differential $z^k\,dz$, for any integer $k$. 
\end{proof}

\begin{remark}
    The integral operators used to replace the projection operators can be expressed as
    \begin{equation*}
        p^{(2)}(z,w) =   \frac{1}{2}\int_{\zeta=0}^{w} B(z,\zeta) 
        - \frac{1}{2}\int_{\zeta=0}^{-w} B(z, \zeta),
    \end{equation*}
    and
    \begin{equation*}
        p^{(3)}(z, w) =   \frac{1}{2}\int_{\zeta=0}^{w} B(z, \zeta)
        + \frac{1}{2}\int_{\zeta=0}^{-w} B(z, \zeta),
    \end{equation*}
    respectively, where $B(z_1,z_2)$ is the normalized canonical bilinear differential of the second kind defined on $\PP^1$ (c.f. \cite{MR0335789}),
    and $w \mapsto -w$ is the unique involutive mapping preserving the spectral curve
    $x = \frac{1}{2}w^2$ around the branch point at $w=0$. In other words, our construction is quite general, and not necessarily restricted to the example at hand.
\end{remark}

To decompose the partition function by genus and marked points, we set
\begin{equation*}
    \tau_1 = e^F,
\end{equation*}
and 
\begin{align*}
    F(u, t) &= \sum_{h=0}^{\infty} u^{h-2} F_{h/2}(t) \\
    F_{g}(t) &= \sum_{n=1}^{\infty} F_{g,n}(t),
\end{align*}
where $F_{g,n}(t)$ is homogeneous of degree $n$ in the variables $t_1, t_2, \ldots$. Then we set
\begin{align*}
    \cU_{0,1} &= -z^2\,dz \\
    \cU_{0,2} &= \tBold \\
    \cU_{\frac{1}{2}, 1} &= \frac{dz}{z} \\
    \cU_{g, n+1} &= \delta F_{g,n+1} \quad
    \text{if $2g-2+n+1 > 0$ and $n, 2g \in \ZZ_{\geq 0}$}.
\end{align*}
After decomposing master equations \eqref{eq:MasterEquationSimple1}, \eqref{eq:MasterEquationSimple2} by degree in both $u$ and $t$,
and replacing the projection operators $\proj^{(i)}$ with the appropriate residue integral,
we arrive at the following recursion relation, valid for all $g, n \geq0$ with $2g-1+n > 0$ and $2g, n \in \ZZ$:
\begin{multline}
    \label{eqn:TopologicalRecursionU}
    \cU_{g,n+1}(z) = 
    \Res_{w\rightarrow 0} \Biggl\{
        K^{(2)}(z, w) \biggl[
                \delta U_{g-1, n+2}(w, w) \\
                + \sum_{\substack{ g_1 + g_2 = g \\ n_1+n_2 = n+2 }}^{\text{no $(g,n+1)$ term}}
                \cU_{g_1,n_1}(w) \cU_{g_2, n_2}(w)
                + \cD_1 \cU_{g-\frac{1}{2}, n+1}(w)
        \biggr] \\
        + K^{(3)}(z, w) \biggl[
                \delta^2\cU_{g-2, n+3}(w, w, w) 
                + \frac{3}{2}\delta \sum_{\substack{g_1 + g_2 = g-1 \\ n_1 + n_2 = n+3}}
                \cU_{g_1, n_1}(w) \cU_{g_2, n_2}(w) \\
                + \sum_{\substack{g_1 + g_2 + g_3 = g \\ n_1 + n_2 + n_3 = n+3 }}^{\text{no $(g,n+1)$ terms}}
                \prod_{i=1}^{3} \cU_{g_i, n_i}(w)
                - \cD_2 \cU_{g-1, n+1}(w)
        \biggr]
    \Biggr\} ,
\end{multline}
where
\begin{equation*}
    K^{(j)}(z, w) = \left(
        (-1)^j \int_{\zeta=0}^{-w} B(z, \zeta) 
        - \int_{\zeta=0}^{w} B(z, \zeta)
    \right)
    \frac{1}{2j (-w^2\,dw)^{j-1}},
\end{equation*}
and
\begin{equation*}
    B(z_1, z_2) = \frac{dz_1 dz_2}{(z_1 - z_2)^2}.
\end{equation*}

Note that this formula is only valid under the convention that
$\delta\cU_{0,2} =  \frac{dz^2}{4z^2}$.

We now construct the correlation functions appearing in topological recursion by setting
\begin{equation*}
    \delta_i = \sum_{j=1}^{\infty} \frac{dz_i}{z_i^{j+1}} \frac{\partial}{\partial t_j}
\end{equation*}
and defining
\begin{equation*}
    W_{g, n}(z_1, \ldots, z_n) = \delta_1 \cdots \delta_n F_{g,n}(t),\quad \text{if $2g-2+n>0$.}
\end{equation*}
The unstable correlation functions are defined as $W_{0,1}(z) = -z^2\,dz$, $W_{\frac{1}{2},1}(z) = \frac{dz}{z}$, and
\begin{align*}
    W_{0,2}(z_1, z_2) &= \delta_2 \cU_{0,2}(z_1) \\
                      &= \sum_{k=1}^{\infty} k \frac{z_1^{k-1}}{z_2^{k+1}} dz_1 dz_2 \\
                      &= \frac{dz_1 dz_2}{(z_1 - z_2)^2}.
\end{align*}
We also need the following renormalized correlation functions:
\begin{equation*}
    \widetilde{W}_{g,n}(z_1, \ldots, z_n) = W_{g,n}(z_1, \ldots, z_n)
    - \delta_{g,0} \delta_{n,2} \frac{dx(z_1) dx(z_2)}{(x(z_1) - x(z_2))^2},
\end{equation*}
where $x = z^2/2$, and
$\delta_{n,m}$ is the Dirac delta function, not the operator $\delta_i$ appearing elsewhere in the paper.
We note in particular, that 
\begin{equation*}
        \widetilde{W}_{0,2}(z, z) = \frac{dz^2}{4z^2},
\end{equation*}
which is exactly the renormalized behavior we need for the topological recursion formula.

To make the formulas a little more digestible, we use the notation $\vec{z} = (z_1, \ldots, z_n)$, 
\begin{multline*}
    \cR^{(2)}W_{g, n+1}(w; \vec{z}) = 
    \widetilde{W}_{g-1, n+2}(w, w, \vec{z}) \\
    + \sum_{\substack{g_1 + g_2 = g \\ Z_1 \sqcup Z_2 = \vec{z}}}^{\text{no $(g,n+1)$ terms}}
    W_{g_1, \abs{Z_1}+1}(w, Z_1) W_{g_2, \abs{Z_2}+1}(w, Z_2),
\end{multline*}
and
\begin{multline*}
    \cR^{(3)}W_{g, n+1}(w; \vec{z}) =
    W_{g-2, n+3}(w, w, w, \vec{z}) \\
    +3 \sum_{\substack{g_1 + g_2 = g-1 \\ Z_1 \sqcup Z_2 = \vec{z}}} W_{g_1, \abs{Z_1}+1}(w, Z_1)
    \widetilde{W}_{g_2, \abs{Z_2}+2}(w, w, Z_2) \\
    + \sum_{\substack{g_1 + g_2 + g_3 = g \\ Z_1 \sqcup Z_2 \sqcup Z_3 = \vec{z}}}^{\text{no $(g, n+1)$ terms}}
    \prod_{i=1}^{3} W_{g_i, \abs{Z_i}+1}(w, Z_i)
\end{multline*}

If we apply the operator $\delta_1 \cdots \delta_n$ to \eqref{eqn:TopologicalRecursionU}, then we find the following
\begin{theorem}
    The correlation functions $W_{g,n}$ for open intersection numbers obey the topological recursion
    formula
    \begin{multline}
        \label{eqn:TopologicalRecursion}
        W_{g,n+1}(z_0, \ldots, z_n) = \Res_{w\rightarrow 0} \Biggl\{
            K^{(2)}(z_0, w) \biggl[
                \cR^{(2)}W_{g, n+1}(w; \vec{z}) + \cD_1 W_{g-1/2, n+1}(w, \vec{z})
            \biggr] \\
            + K^{(3)}(z_0, w) \biggl[
                \cR^{(3)}W_{g, n+1}(w; \vec{z}) - \cD_2 W_{g-1, n+1}(w, \vec{z})
            \biggr]
    \Biggr\},
    \end{multline}
   with initial conditions
   \begin{align*}
       W_{0,1}(z) &= -z^2\,dz \\
       W_{\frac{1}{2},1}(z) &= \frac{dz}{z} \\
       W_{0,2}(z_1, z_2) &= B(z_1, z_2),
   \end{align*}
   where
   \begin{align*}
       \cD_1 &= dz \left(-\frac{d}{dz} + \frac{1}{z}\right) \\
       \cD_2 &= \frac{dz^2}{2} \left( \frac{d^2}{dz^2} - \frac{3}{z}\frac{d}{dz} + \frac{3}{z^2} \right).
   \end{align*}
   
\end{theorem}

\section{A conjectural refinement}
\label{sect:Q-Grading}

One issue that arises in the above formulation of topological recursion for open intersection numbers is that the correlators combine intersection numbers from a union of disjoint moduli spaces. In particular, a given function $F_{h, n+1}$, for $2h, n \in \ZZ_{\geq 0}$, has contributions from moduli spaces of genus $g$ curves with $b$ boundary components for all $g, b$ with $g + b/2 = h$.
If the contribution to $F_{h,n}$ from genus $g$ surfaces with $b$ boundary components is $F_{g,b,n}$, it is natural to try to introduce a parameter $Q$, and write
\begin{equation*}
    F_{h,n}(t, Q) = \sum_{\substack{g \in \ZZ \\ 0 \leq g \leq h}} Q^{2(h-g)} F_{g, 2(h-g), n}(t),
\end{equation*}
as well as the associated correlators
\begin{equation*}
    W_{g, k, n} = \delta_1 \cdots \delta_n F_{g, k, n}.
\end{equation*}
If such a decomposition is possible, 
we can resum to obtain correlators with only integral genus parameter
\begin{equation*}
    \Omega_{g,n} = \sum_{k=0}^{\infty} Q^k W_{g, k,n}.
\end{equation*}
This form of the correlators is required if one were to attempt to find a spectral curve whose correlators under the standard theory of topological recursion calculate open intersection numbers.

It is tempting to try to insert the $Q$ grading into the topological recursion formula \eqref{eqn:TopologicalRecursion} by replacing $W_{\frac{1}{2},1} \mapsto Q W_{\frac{1}{2},1}$ and $\cD_i \mapsto Q^i \cD_i$, as this would produce correlation functions with terms of correct degree in $Q$.
Unfortunately, this proves unsuccessful as these $Q$-graded correlators
$W_{g,n}(Q; \vec{z})$ with $2g-2 + n \geq 3$ are not symmetric.
However, correlators with $2g-2 + n < 3$ match exactly to the correlators coming from $\tau_Q$, the Kontsevich-Penner matrix model.
This motivates the following
\begin{conjecture}
    \label{conj:QGrading}
    Let $\tau_Q = e^{F_Q}$, with $\tau_Q$ given by the Kontsevich-Penner matrix model \eqref{eqn:MatrixIntegral}. Then the correlators $W_{g,n}(Q; \vec{z}) = \delta_1 \cdots \delta_n F_{g,n}(Q; t)$ given from the expansion
    \begin{equation*}
        F_Q(t) = \sum_{g,n} u^{2g-2}F_{g,n}(Q;t)
    \end{equation*}
   are the $Q$-graded correlators for open intersection numbers. 
\end{conjecture}
The conjecture is true for $g = 0, \frac{1}{2}, 1$. In addition, the correlators of all augmented genera for $F_Q$ 
exhibit the proper degree behaviour in $Q$.
One difficulty that arises is that moduli spaces with different numbers of boundary components can have common boundary using the compactification of Solomon and Tessler.
Moreover, the boundary contributes non-trivially to the intersection numbers, realized in the form of nodal ribbon graphs contributing to the total in Tessler's combinatorial model \cite{Tessler:2015ys}.

However, what seems to hold true, at least in the low genus examples that can be calculated by hand, is that the ribbon graphs naturally sort themselves into neat piles, each contributing to a term with fixed degree in $Q$, and the resulting expression matching the one predicted by the conjecture. So, while the evidence in support of Conjecture~\ref{conj:QGrading} is hardly definitive, it certainly seems promising enough to warrant further investigation.

\section{Quantum curve equation}
\label{sect:QuantumCurve}

In this section we derive a quantum curve for open intersection numbers.
In topological recursion, the quantum curve is obtained from the spectral curve via \emph{quantization}, whereby we replace $y$ with $\hbar \frac{d}{dx}$. Then, in many cases, and with the proper choice of ordering of the now non-commuting variables, one obtains the quantum curve equation
\begin{equation*}
    P(\hat{x}, \hat{y})e^{\Psi} = 0,
\end{equation*}
where $\Psi$ is the \emph{principal specialization} of the partition function. 
In our case, it can be realized by substituting 
\begin{align*}
    \tilde\Psi_Q &= 
    F\biggr|_{\substack{u \mapsto \hbar \\
    t_k \mapsto \frac{\hbar}{k z^k}}} \\
    \Psi_Q &= \tilde\Psi_Q + \frac{z^3}{3} - \frac{3}{4}\log \frac{z^2}{2}.
\end{align*}
It also corresponds with taking $N=1$ in the matrix integral \eqref{eqn:MatrixIntegral}.
Hence, by work of Brezin and Hikami \cite{MR2874239},
we have the quantum curve equation
\begin{theorem}
    If $\Psi_Q$ is the principal specialization of $F_Q$ and
    $x = \frac{1}{2}z^2$, then
\begin{equation*}
    \left(\hbar^3 \frac{d^3}{dx^3} - 2\hbar x \frac{d}{dx}
    + 2\hbar(Q-1)\right) e^{\Psi_Q} = 0.
\end{equation*}
The semi-classical limit is, for all values of $Q$,
\begin{equation*}
    y^3 - 2xy = 0.
\end{equation*}
\end{theorem}

We note in particular that the spectral curve is reducible and has degree 3.
Since there is currently no known way of calculating topological recursion for a general reducible spectral curve, it is not surprising that the formulas we obtain are not in exact correspondence with the standard topological recursion.
Whether or not this example can be generalized to other reducible curves is a topic for future work.
However, the fact that the curve is degree 3 does help explain why the formula we obtain most closely matches topological recursion in the rank 3 case.


\bibliographystyle{plain}
\bibliography{references}

\end{document}